\documentclass[journal]{IEEEtran}
\ifCLASSINFOpdf
\else
\fi
\usepackage{comment}
\usepackage{booktabs}
\usepackage{balance}
\usepackage{graphicx}
\usepackage{amsmath,amssymb,amsfonts,amsthm}
\usepackage{xcolor}
\usepackage[ruled]{algorithm2e}
\SetKwComment{Comment}{/* }{ */}

\usepackage{mathtools}

\theoremstyle{definition}
\newtheorem{problem}{Problem}
\theoremstyle{definition}
\newtheorem{definition}{Definition}
\newtheorem{theorem}{Theorem}

\usepackage[symbol]{footmisc}

\usepackage{subfig}

\usepackage{float}
\newfloat{Problem}{htbp}{loe}
\floatname{Problem}{Opt. Problem}

\begin{document}
%
\title{
Sparing User Time with a Socially-Aware Independent Metaverse Avatar
}
%
%
%

\author{Theofanis P. Raptis, Chiara Boldrini, Marco Conti, and Andrea Passarella
\thanks{All authors are with the Institute of Informatics and Telematics (IIT) of the National Research Council (CNR), Italy. email: first.last@iit.cnr.it}
\thanks{This work was partially supported by SoBigData.it, which receives funding from the European Union -- NextGenerationEU -- National Recovery and Resilience Plan (Piano Nazionale di Ripresa e Resilienza, PNRR) – Project: ``SoBigData.it – Strengthening the Italian RI for Social Mining and Big Data Analytics'' -- Prot. IR0000013 -- Avviso n. 3264 del 28/12/2021. Chiara Boldrini was supported by PNRR - M4C2 - Investimento 1.4, Centro Nazionale CN00000013 - ``ICSC -- National Centre for HPC, Big Data and Quantum Computing'' -- Spoke 6, funded by the European Commission under the NextGeneration EU programme. Marco Conti and Andrea Passarella were partly supported by PNRR -- M4C2 -- Investimento 1.3, Partenariato Esteso PE00000013 -- ``FAIR - Future Artificial Intelligence Research'' -- Spoke 1 ``Human-centered AI'', funded by the European Commission under the NextGeneration EU programme.}
}
\maketitle
\begin{abstract} 
The Metaverse is redefining digital interactions by merging physical, virtual, and social dimensions, yet its effects on social networking remain largely unexplored. This work examines the role of \emph{independent avatars} (autonomous digital entities capable of managing social interactions on behalf of users), to optimize social time allocation and reshape Metaverse-based Online Social Networks.  
We propose a novel computational model that integrates a quantitative and realistic representation of user social life, grounded in evolutionary anthropology, with a framework for avatar-mediated interactions. Our model quantifies the effectiveness of a partial replacement of in-person interactions with independent avatar interactions. Additionally, it accounts for social conflicts and specific socialization constraints. We leverage our model to explore the benefits and trade-offs of an avatar-augmented social life in the Metaverse. Since the exact problem formulation leads to an NP-hard optimization problem when incorporating avatars into the social network, we tackle this challenge by introducing a heuristic solution.
%
Through simulations, we compare avatar-mediated and non-avatar-mediated social networking, demonstrating the potential of independent avatars to enhance social connectivity and efficiency. Our findings provide a foundation for optimizing Metaverse-based social interactions, as well as useful insights for future digital social network design.

\end{abstract}


%
\IEEEpeerreviewmaketitle

\section{Introduction}

The Metaverse is a digitally constructed environment~\cite{HWANG2022100082} that has been envisioned in various ways over time but has recently attracted significant attention due to its potential to reshape human interactions in virtual and augmented spaces. It represents a highly interconnected ecosystem that blends physical, digital, and social elements, enabling users to engage in immersive experiences beyond what traditional online platforms offer. As an evolution of the Internet and social media, the Metaverse introduces a more interactive and dynamic way for individuals and entities to communicate and collaborate~\cite{10376184}. One of the defining aspects of the Metaverse is its ability to push the boundaries of cyber-physical connectivity, opening up new possibilities for social engagement and redefining the way people interact online~\cite{10368182}. Within this evolving digital landscape, \textbf{avatars and agents} serve as key components~\cite{9944868}. An \emph{avatar} 
is a virtual representation of a user, typically in the form of a graphical or 3D model, that allows individuals to navigate and interact within the Metaverse. Avatars act as extensions of their users, reflecting their identity and actions. In contrast, an \emph{agent} 
is an autonomous digital entity driven by artificial intelligence or scripted behaviors. Unlike avatars, agents can function independently, making decisions and executing actions without direct human control.  
The study of these virtual entities is gaining traction in social sciences, where researchers are exploring their impact on human experiences in the Metaverse. By examining how avatars and agents shape interactions, researchers aim to enhance the design and functionality of these digital spaces, making them more engaging and meaningful for users.
%
%
%
%
%
These new forms of interaction could significantly reshape online social networks, making them more dynamic and adaptable. 

We aim to investigate the unprecedented potential of this transformation, focusing particularly on how autonomous digital entities influence social networking dynamics in the Metaverse.
While discussions on the Metaverse are gaining traction, its computational social networking aspects remain largely unexplored (see discussion in Section~\ref{sec:relwork}).
%
Our work focuses on a previously unexplored dimension: the role of \emph{independent avatars} in augmenting users'  social networking capacity. We specifically investigate how these autonomous avatars can extend a user’s capacity for social engagement, thereby expanding the structure of their social network in the Metaverse.
Historically, avatars have been perceived as extensions of human users, manually controlled in real time to facilitate digital interactions~\cite{Sibilla_Mancini_2018}. However, in the Metaverse, we propose an expanded definition in which avatars transcend their traditional roles, operating autonomously when not under direct user control. We call these entities \emph{independent avatars}. These digital counterparts are capable of \textbf{navigating social interactions on behalf of users}, exhibiting agent-like behaviors that allow them to independently manage certain aspects of online socialization.
By delegating part of their social interactions to autonomous avatars, users could maintain their existing social ties with greater efficiency while freeing up their own time for new engagements. This redistribution of social effort could effectively expand users' overall socialization capacity, leading to significant changes in the structure and dynamics of Metaverse-based social networks~\cite{aral2023exactly}. 

In order to study this effect, this work proposes a novel computational model of how independent avatars could assist users in managing their social interactions. Our approach leverages quantitative models from social science such as \emph{ego networks}~\cite{DUNBAR201539} and models from human-centric Virtual Reality (VR) such as \emph{social presence}  \cite{10.3389/frobt.2018.00114}, 
and \emph{social cueing} \cite{doi:10.1080/10447318.2023.2193514} to build a microscopic characterization of social networking in the Metaverse era. 
By incorporating such concepts in our computational model, we aim to not only capture the complexities of users' social interactions but also provide a robust foundation for understanding how independent avatars can effectively navigate and augment the user's social time capacity, while allowing users to still control and be updated about the social interactions the avatars carry out.

The key contributions of the paper are the following.
\begin{itemize}
    \item We propose a novel computational model that combines a quantitative and realistic representation of user social life, grounded in evolutionary anthropology, with a framework for avatar-mediated interactions. This framework quantifies the effectiveness of avatar interactions relative to in-person interactions while also accounting for social conflicts and specific socialization constraints.
    \item We leverage this model to investigate the benefits and trade-offs of an avatar-augmented social life in the Metaverse.
    \item Since we prove that the exact problem formulation leads to an NP-hard optimization problem when incorporating avatars into the social network, we introduce a heuristic solution that aims at addressing the problem in a sorted social request scheduling manner.
    \item We conduct a holistic (simulative) performance evaluation of our heuristic solution, demonstrating that the use of independent avatars can significantly reduce social cost, improve fairness in social time allocation, and enhance system robustness under tight constraints. Our analysis spans multiple ego network sizes, conflict densities, and system parameters, revealing that substantial benefits can be achieved even with limited avatar availability.
\end{itemize}

The structure of this paper is as follows. In Section~\ref{sec:relwork}, we discuss the relevant related literature. In Section \ref{sec:model}, we introduce the reader to the modeling of ego networks, time fragmentation within a year, non-avatar-mediated and avatar-mediated socializing time capacities and debriefing times, social presence and cues, as well as social request deadlines and cost functions. In Section~\ref{sec:prob}, we formulate the problem of the optimization of the user's spare time and we prove that it is NP-hard. Following, in Section~\ref{sec::algo}, we design an efficient heuristic approximation of the optimal spare time, focusing on a spare time-centric process with sorted sequential request scheduling. In Section~\ref{sec::eval}, we present the simulation experiments we conducted to assess the impact of the introduction of an independent avatar to the social setting via the usage of our heuristic. Finally, in Section~\ref{sec::conc}, we conclude the paper.


\section{Related works}
\label{sec:relwork}

\subsection{Ego networks}

The \emph{social brain hypothesis} from anthropology~\cite{dunbar1998social} suggests that the number of stable social relationships a primate can maintain is limited by the size of its neocortex. For humans, this cognitive limit is approximately 150 individuals (known as \emph{Dunbar's number}) which reflects the mental effort required to sustain meaningful, high-quality relationships. These 150 connections exclude casual acquaintances and refer instead to emotionally significant ties, typically maintained through regular contact (e.g., birthday or holiday greetings).
These relationships are organized in concentric layers of increasing intimacy~\cite{hill2003social,Zhou2005}, as illustrated in Figure~\ref{fig:egonet}. The innermost layer, the \emph{support clique}, includes close family and best friends. This is followed by the \emph{sympathy group} (those whose death would cause deep grief), the \emph{affinity group} (friends, colleagues, and extended family), and the \emph{active network} (people contacted at least once per year). Some individuals also maintain an additional innermost layer, nested within the support clique, comprising on average 1.5 people with whom they have especially strong emotional bonds~\cite{Dunbar2015}.
Although originally developed for offline interactions, this model has been validated in digital communication as well. Dunbar's number and the layered structure have been observed in email~\cite{Haerter2012}, mobile phone calls~\cite{Miritello2013}, and online platforms such as Twitter and Facebook~\cite{gonccalves2011modeling,Dunbar2015}. In this work, we adopt the \emph{ego network model} as our reference framework for representing an individual’s surrounding social structure.

\vspace{-10pt}
\subsection{Social networking in the Metaverse}

While discussions on the Metaverse are gaining traction, its computational social networking aspects remain largely unexplored. Most existing studies have concentrated on virtual reality infrastructure and capabilities, without addressing the impact of novel Metaverse-specific elements—such as avatars and agents—on the structure of users’ social networks. For instance, research in~\cite{ARPACI2022102120} examines social sustainability in the Metaverse through a model based on personality traits, while~\cite{OH2023107498} explores whether the increased social presence in the Metaverse enhances supportive interactions among young users. In another line of work, \cite{10.1145/3511808.3557487} formulates the problem of user co-presence and occlusion-aware recommendation in virtual social settings, focusing on how technological advancements influence social presence. Similarly,~\cite{10.1145/3517745.3561417} systematically evaluates network performance across five major social virtual reality platforms, revealing fundamental scalability challenges that currently hinder the realization of a fully functional social Metaverse.

Historically, avatars have been perceived as extensions of human users, manually controlled in real time to facilitate digital interactions~\cite{Sibilla_Mancini_2018}. Aspects such as avatar creation \cite{BARTA2024108192}, security \cite{sandeep2025}, privacy \cite{Eltanbouly16022025} and human-computer interaction \cite{10086600} have been studied in the related literature. In our work, however, we propose an expanded definition in which avatars transcend their traditional roles, operating autonomously when not under direct user control. We call these entities \emph{independent avatars}. These digital counterparts are capable of \textbf{navigating social interactions on behalf of users}, exhibiting agent-like behaviors that allow them to independently manage certain aspects of online socialization.


Building upon the simplified model in~\cite{10637545}, this work introduces a more realistic representation of time constraints. Specifically, user encounters are now restricted to predefined time windows, and any delay incurs a social cost. This better reflects real-world social dynamics, such as spouses meeting for an anniversary, colleagues discussing urgent matters, or friends coordinating time-sensitive gatherings.
Additionally, we model mutually exclusive encounters—situations where attending one social event precludes participation in another. For instance, if a user goes to dinner with a friend, they cannot simultaneously meet a third person who is disliked by either party. This added realism significantly increases the complexity of the original model from~\cite{10637545}, which was formulated as a simple linear programming optimization. The new problem formulation is NP-hard, requiring a more sophisticated approach for its resolution.
To address this challenge, we propose a heuristic solution. The design and analysis framework from~\cite{10637545} remains a foundational element of our modeling approach and plays a crucial role in guiding the algorithmic design. 

\section{System model}
\label{sec:model}


In this section, we introduce the key elements of our model: (i) a quantitative framework for user’s social interactions (Sec.\ref{sec:model_social}) and (ii) a quantitative framework for avatar-mediated social interactions in the Metaverse (Sec.\ref{sec:model_avatar}). Then, in Section~\ref{sec:model_numbers}, we ground the parameters of our system model in realistic values derived from the related literature.

For (i), we build upon the ego network model from evolutionary anthropology, using it to compute sociality volume—the total time a user spends with their meaningful social relationships- and its distribution among the different alters. The ego network model provides a compact, high-level representation of a user’s social structure. However, in daily life, social relationships are dynamically constrained by factors such as specific timing requirements (e.g., spending time with family during Christmas in Western societies). To better capture these complexities, we complement the ego network model with a realistic and flexible representation of individual interaction sequences, where each social encounter is assigned a deadline by which an interaction should occur. To account for situations where maintaining multiple relationships on the same day is impractical (e.g., due to geographical separation or social conflicts), we introduce a conflict graph—a compact representation of mutually exclusive interactions within the user’s social network.
For (ii), we incorporate concepts from human-centric VR to model avatar-mediated interactions. Specifically, we consider social presence—the subjective experience of being with a ``real'' user and having access to their thoughts and emotions—and social cues, which include verbal and non-verbal signals used in communication. These elements allow us to quantify the effectiveness of avatar-mediated interactions compared to direct user interactions and to properly ``weigh'' the social impact of avatar interactions relative to in-person encounters. Specifically, this allows us to determine, for a given ego-alter relationship, the time required by an avatar to interact with the alter, such that the interaction is equally effective as if it were done by the ego themself.
Additionally, we account for the need for debriefing, where the user must synchronize with social developments after an avatar-mediated interaction. This ensures that the user remains updated on their relationships and is prepared for future direct interactions with the same alter.

Given (i) and (ii), we provide a framework to optimise the use of the ego social capacity - exploiting the presence of an avatar - such that the ego can maximise their spare capacity without compromising the effectiveness of their relationships with alters. We show how this spared time can be ``invested'' in establishing new relationships, ultimately expanding the ego network.

\subsection{A model for the user social life}
\label{sec:model_social}

\subsubsection{User's ego network}

\begin{figure}[t!]
\begin{center}
\includegraphics[width=0.4\columnwidth]{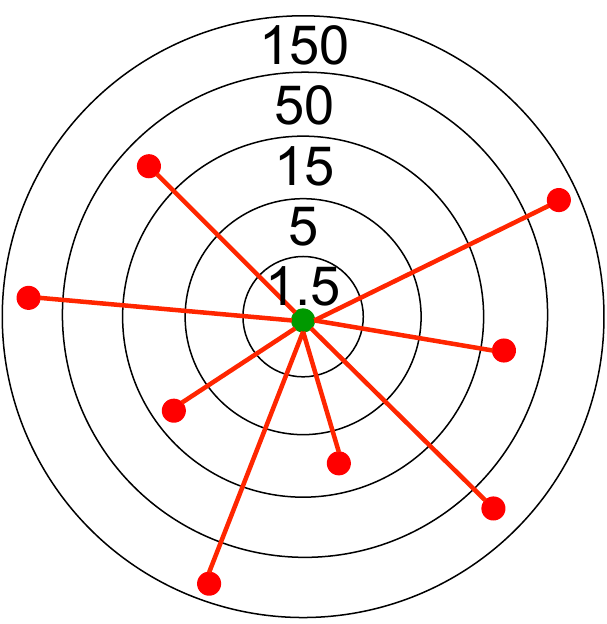}
\caption{Layered structure of human ego networks. The green node represents the ego, the red nodes represent its alters. The numbers give the size of the corresponding social circle.}\vspace{-10pt}
\label{fig:egonet}
\end{center}
\end{figure}

Following the Dunbar's model of ego networks (Figure~\ref{fig:egonet}), we assume that a user $u$ (the ego, in Dunbar's terminology) maintains a personal network that consists of the people $v_1, v_2, ..., v_n \in V$ (called alters) with whom $u$ interacts with regularly. 
Within this ego network, $u$ has different levels of closeness to each person. At the core of the network are a few people with whom $u$ has the closest relationships, while at the periphery there are acquaintances that $u$ only interacts with occasionally.
In an ego network, relationships are organized, according to their strength, in concentric circles (or layers). Typically, the ego network of $u$ consists of 4-5 layers.

User $u$ may face limitations in terms of time and resources for socializing with other individuals. A major challenge is that $u$ only has a limited amount of time (and cognitive capacity). 
This may mean that $u$ must prioritize certain connections over others, and may not be able to invest in as many new relationships as they would like. Another challenge is that expanding the network may require the ego to navigate unfamiliar social contexts or communities, which can be intimidating or uncomfortable. It may take time and effort to build trust and establish rapport with individuals who are not part of the ego's existing social circles. Finally, there is also the risk of spreading oneself too thin by trying to maintain connections with too many individuals. 
We assume that each alter $v \in V$ in the ego network of user $u$ requires a certain amount of socialization time $\tilde{x}_v$ with $u$. Note that the time spent with each alter generally varies from alter to alter. In principle, node $u$ should spend $\tilde{X}$ time (non-avatar-mediated) socializing, where

\begin{equation} \label{eqn:Tu}
\tilde{X} = \tilde{x}_1 + \tilde{x}_1 + .... + \tilde{x}_n.
\end{equation}
We refer to $\tilde{X}$ as \textbf{baseline socialization capacity}. A user $u$ typically has sufficient socialization capacity to maintain their ego network. 
 However, there may be circumstances where this capacity is reduced due to exogenous factors, such as illness or an exceptionally busy period at work. To ensure the flexibility of our model, we assume that each user has an \textbf{actual socialization capacity}, denoted by $\tilde{X}'$ (expressed in units of time per year), which ideally meets the threshold $\tilde{X}$, though this is not always guaranteed. Thus, it holds that 
 \begin{equation}
     \tilde{X}' \leq \tilde{X}.
 \end{equation}
%


\subsubsection{Modelling time and and inter-alter conflicts}

We consider a fixed time period, such as a year. The total socialization time for each alter can be fragmented and distributed across $k$ ordered time slots $d_i$ within this period (e.g., days of a year), where each slot has a specific duration $s(d_i)$.  
To model conflicts in physical socialization scheduling, arising from geographical constraints or other limitations, we introduce a conflict graph $G_c = (V_c, E_c)$. In this graph, each node represents an alter, and an edge $(i, j) \in E_c$ indicates that alters $i$ and $j$ cannot be assigned to the same time slot due to conflicting constraints.


\subsubsection{Social obligations and the cost of a delayed interaction}

The relationship between the user and an alter $v$ is made up of several social encounters, each with a certain duration. 
Each alter $v$ comes with a set of predefined social requests with deadlines on the time slots of the user. The social requests of alter $v$ are formulated as a list $R_v$ of triplets $(r^j_v, d'^j_v, d''^j_v)$, in which $r^j_v$ corresponds to the size (in terms of time) of the request $j$ of alter $v$, $d'^j_v$ the time slot that this request can start being served, and $d''^j_v$ the deadline until this request can be served without a socialization cost. 
A request $j$ that is served outside the time window $[d'^j_v, d''^j_v]$ receives a socialization penalty, according to a cost function $c(d''^j_v, d)$, where $d$ is the time slot that it was eventually served.

\subsection{A model for avatar-mediated social interactions}
\label{sec:model_avatar}


We assume that user $u$ owns an independent avatar $a$ that can operate in the Metaverse also independently from $u$ (e.g., when $u$ is offline or otherwise occupied). 
Given that avatar $a$, when operating independently in the Metaverse, can contribute a fixed amount of socializing time $Y$ (avatar-mediated) to the current activities of $u$, we assume that $a$ can virtually augment $u$'s socializing time and therefore
enhance $u$'s socializing time across the existing alters in the ego network\footnote{Example: Consider a user $u$ who has limited daily time for socializing but wants to maintain active connections with several friends. Without an avatar, $u$ must personally schedule and attend every interaction. With an avatar $a$, $u$ can delegate some of these interactions: For example, $a$ might check in with a friend $v$ in a virtual café, exchange brief updates, and later summarize the key points to $u$ during a short debriefing session. We assume that only $u$ owns an avatar and that the interaction remains meaningful even if it is one-sidedly delegated. Avatar-to-avatar socialization, while possible, is not considered in our model. This setup enables $u$ to scale social engagement while retaining cognitive involvement through bounded debriefing.},
partitioning its time in a non-uniform manner, with\footnote{It is therefore apparent that during the activity delegation from $u$ to $a$, $u$ might decide to reduce the amount of non-avatar-mediated time for alter $v$ and replace it with an amount of avatar-mediated time instead.}
\begin{equation}\label{eqn:Ta}
Y = y_1 + y_2 + .... + y_n.
\end{equation}

In order for $u$ to stay up-to-date with $a$'s activities and interactions, $u$ will need to sync with $a$ periodically to receive a \emph{debrief} of its socialization updates for each $v$. This synchronization is crucial to maintain $u$'s cognitive engagement in these social relationships. According to the ego network model, without this cognitive involvement, relationships may transition out of the ego network and become mere acquaintances, as $u$ would not spend \emph{any} cognitive resources on those relationships, which would inevitably fade out with time. 
Thus, despite the potential resource costs, the debriefing process is necessary for staying informed and engaged with the $a$'s activities, as well as effectively incorporating the social presence value achieved by $a$. In order for $u$ to get debriefed by $a$ for all alters in $V$, $u$ would need to spend 
    \begin{equation}
    Z = z_1 + z_2 + ... + z_n \label{eqn:Td}
\end{equation}
time (referred to as \textbf{debriefing time Z}). However, due to the fact that attention span (defined as the amount of time spent concentrating on a task before becoming distracted \cite{berger2018present}) can decrease in some given amount of time, leading to distractibility, we assume that the attention during the debriefing process is eventually uncontrollably diverted to another activity. We therefore impose an upper threshold of $Z_{\text{max}}$ on the debriefing time, with 
\begin{equation}
  Z \leq Z_{\text{max}}.  \label{eq:zmax}
\end{equation}
This threshold incurs the potential unavailability of $u$ to effectively get debriefed $\forall v \in V$. For simplicity in our analysis, we assume that the debriefing of an avatar socialization session must happen on the same day of the year as the day that the socialization session took place.


\subsubsection{Weighing avatar-mediated encounter by social presence}

In order to formulate the avatar socializing capacity $Y$, we need to better understand the social presence capacity of an independent avatar $a$ in the Metaverse. Social presence is broadly defined as the subjective experience of being present with a ``real'' user and having access to her thoughts and emotions. 
Different forms of networked communication systems offer different levels of social presence. 
Therefore, in order to achieve the same social presence value by using two different networked communication means, a user might need to spend different amounts of time. Social presence is often evaluated under different contexts (for example, user-user interaction, user-avatar interaction), and as such, some diversity in measures cannot be avoided. 
In the context of our paper, it is necessary to understand how different technological features of the Metaverse networked communication influence perceptions of social presence so as to be able to define a social presence function~$f_{s}$. Defining the function will, in turn, enable us to characterize the relation between the avatar-mediated communication time $y$ and the non-avatar-mediated communication time $\tilde{x}$.

In the social psychology literature, our notion of ``independent avatar'' would naturally fall in the category of ``agent''. Research works that have investigated the impact of the (perceived) agency on social presence typically assume the user as an actual person or a computerized character prior to the interaction. As the survey \cite{10.3389/frobt.2018.00114} nicely points out, the majority of the related literature converges to the conclusion that avatars generally elicit greater social presence value than agents, in the sense that users typically feel higher levels of social presence when a virtual counterpart was thought to be controlled by an actual person rather than by a computer program.

Digital media vary in their ability to transmit social cues and thereby facilitate social presence 
through technology-mediated interpersonal communication. Media that are high in social presence are likely better at facilitating social connectedness because they are closer to face-to-face communication compared to media that are lower in social presence. We naturally assume that non-avatar-mediated communication can ideally achieve a higher social presence value compared to avatar-mediated communication. Therefore the social presence value achieved with a unit of time spent in the first case ($\tilde{x}_v$) does not equal the value with a unit of time in the second case ($y_v$), and we need a mapping $f_{s}$ to convert the social presence of the avatar to the equivalent social presence of the user. In other words, in order for $u$ to have achieved the same amount of social presence value with a target alter $v$ when using an independent avatar $a$, $a$ would have needed to spend $y_v = f_{s}(\tilde{x}_v)$ time with alter $v$. According to \cite{2012appel}, $f_{s}$ can be modelled as a linear function of $\tilde{x}_v$:
\begin{equation}
    y_v = f_{s}(\tilde{x}_v) = \beta_v \cdot \tilde{x}_v,
    \end{equation}
where the proportionality factor $\beta_v$ is the \textbf{social presence conversion factor}. $\beta_v$ represents the factor by which the time spent by the avatar with alter $v$ should be increased, relative to the time spent by the ego, in order to compensate for the reduced social presence of the avatar. We discuss in Section~\ref{sec:model_numbers} how to set its value. 

\subsubsection{Quantifying the debriefing time with social cues and compression}

We now better formulate the debriefing process, which requires both time and attention from $u$. Depending on the frequency and complexity of $a$'s social interactions, the debriefing process may take varying amounts of time for each $v \in V$ and may need to be prioritized alongside other activities. For each $v \in V$, we assume that the required amount of time for $u$ to get debriefed is not more than the actual interaction time between $a$ and $v$ and that $z_v$ is a function of $y_v$, therefore:
\begin{equation}
z_v = f_{d}(y_v) \leq y_v.
\end{equation} 
The individual debriefing efficiency can be quantified through compression efficiency 
and avatar anthropomorphic and social cue measurements 
which can be expressed as follows:
\begin{equation}
z_v = c \cdot \delta \cdot y_v = \gamma \cdot y_v,
\label{eq:debriefing_efficiency}
\end{equation}
%
%
with $c$ being the \textbf{compression ratio}, $\delta$ being the \textbf{social cuing efficiency} that $a$ can achieve for debriefing. The compression ratio captures the fact that regular human interactions are interspersed with downtime in the conversation and are both partially inefficient (e.g., due to momentary misunderstandings) and redundant (because we tend to repeat concepts multiple times). Thus, a 1-hour conversation can often be summarised in a much shorter time. The compression ratio expresses this shrinking of time when summarising the salient points of a conversation. The social cuing efficiency $\delta$ captures the efficiency of an avatar to convey a certain message depending on its level of anthropomorphism and its ability to reproduce social cues (subtle interpersonal information, such as tone, gestures, or emotional nuance). A higher $\delta$ indicates an avatar that is more anthropomorphic or socially expressive, making the summary not only concise but also richer in communicative fidelity.


\subsection{Grounding our model variables}
\label{sec:model_numbers}

As with any computational model, we have defined multiple variables to capture different aspects of the system being modeled. Building on the related literature, we now discuss how these variables can be assigned realistic values to ensure the model's validity and applicability.

Let us begin by considering the socialization time available to the ego and the avatar. As a human, the ego’s daily life is naturally divided among various activities such as work, socialization, and other responsibilities. Consequently, only a fraction of the total 8,760 hours in a year can realistically be devoted to social interactions. According to~\cite{dunbar1998theory}, individuals spend approximately 20\% of their time on socialization. Therefore, the expected amount of time available for non-avatar-mediated social interactions is:
\begin{equation}
\mathbb{E}[\tilde{X}] = 8760 \cdot 0.2 = 1752.
\end{equation}
Vice versa, we can assume that the avatar could potentially spend all of its available time socializing, so, for a given 8,760 hours in a year, the value of non-avatar-mediated socializing time would be
\begin{equation}
Y \leq 8760.
\end{equation}
Thus, the avatar has more time to devote to socialization than the average human. However, as discussed before, its interactions with humans are less effective than interactions among humans. This efficacy is measured by the social presence conversion factor $\beta_v$. 
The expected $\beta_v$ can be estimated based on the findings reported in~\cite{2012appel}. In that study, participants rated social presence using a 5-point Likert scale. The average score for human–human (HH) interactions was 4.182, while human–computer (HC) interactions received an average score of 3.236 under the same experimental conditions. The ratio between these two values provides a conversion factor for translating the efficacy of HH interactions into the HC domain:
\begin{equation}
\mathbb{E}[\beta_v] = \frac{4.182}{3.236} \approx 1.29.
\end{equation}
In other words, replacing an HH interaction with an HC interaction would require the latter to last approximately 1.29 times longer to achieve an equivalent level of perceived social presence.


We now focus on the debriefing process between the ego and avatar $a$. We used three parameters: $Z_{max}$ captures the upper bound on the debriefing time determined by the attention span of the ego. $c$ and $\delta$ capture, together, the efficacy of the communication between the ego and the avatar. Let us start with $Z_{max}$. According to \cite{doi:10.1152/advan.00109.2016}, the threshold is 50 minutes per day, then, for a given year, we would have a debriefing time of
\begin{equation}
Z_{\text{max}} \approxeq 304 \textnormal{ hours}.
\end{equation}
We now focus on the compression ratio and social cueing efficiency, which are used to model the debriefing process between the ego and the avatar. According to known works on compression ratios, such as  \cite{10.1145/312624.312665} and \cite{ghalandari-etal-2022-efficient}, we can set 
\begin{equation}
\mathbb{E}[c] = 0.54.
\end{equation}
This implies that the summary of an average interaction can be delivered in only half the time of the original interaction.
Social cues can be measured in several directions (verbal, visual, auditory, etc.). \cite{doi:10.1080/10447318.2023.2193514} provides a quantitative survey-based assessment that captures both the visual and verbal efficiency of chatbots, using a 7-point Likert scale. As an example, in the best case, highly anthropomorphic chatbots with a human-like conversation style reach a score of~6 out of~7. We use these results to assign a value to $\delta$. Specifically, in the above example, we say that 
\begin{equation}
\mathbb{E}[\delta] = 7/6.
\end{equation}
Note also that $\delta$ is fixed for every ego, as it simply depends on the social presence properties of the avatar associated with the ego. Circling back to Equation~\ref{eq:debriefing_efficiency}, this means that $\gamma = c \cdot \delta$ can be set to $0.63$. 
For the convenience of the reader, the notation presented in this section is summarised in Table~\ref{tab:notation}.

\begin{table}[t!]
\centering
\caption{Summary of notation}
\label{tab:notation}
\begin{tabular}{@{}ll@{}}
\toprule
\textbf{Symbol} & \textbf{Description} \\
\midrule
$u$ & Ego user (central node of the ego network) \\
$v \in V$ & Alter (person in the ego network) \\
$\tilde{x}_v$ & Baseline (non-avatar-mediated) socializ. time with alter $v$ \\
$\tilde{X}$ & Baseline total socialization capacity $\sum_v \tilde{x}_v$ \\
$\tilde{X}'$ & Actual available socialization capacity ($\leq \tilde{X}$) \\
$d_i$ & Time slot $i$ (e.g., a day) \\
$s(d_i)$ & Duration of time slot $d_i$ \\
$G_c$ & Conflict graph of alters\\
$R_v$ & List of social requests of alter $v$ \\
$r^j_v$ & Duration of the $j$-th social request from alter $v$ \\
$d'^j_v$ & Earliest start time slot for $r^j_v$ \\
$d''^j_v$ & Deadline time slot for $r^j_v$ \\
$c(d''^j_v,d)$ & Social cost if the request is served at day $d$ after its deadline \\
$a$ & Independent avatar of user $u$ \\
$y_v$ & Avatar-mediated interaction time with alter $v$ \\
$Y$ & Total avatar-mediated socialization time $\sum_v y_v$ \\
$Z$ & Total debriefing time by user $u$ for avatar interactions $\sum_v z_v$ \\
$Z_{max}$ & Maximum allowed debriefing time (attention span limit) \\
$f_s(\cdot)$ & Social presence conversion function mapping $\tilde{x}_v$ to $y_v$ \\
$\beta_v$ & Social presence conversion factor ($y_v = \beta_v \tilde{x}_v$) \\
$f_d(\cdot)$ & Debriefing function mapping $y_v$ to $z_v$ \\
$z_v$ & Debriefing time for alter $v$ \\
$c$ & Compression ratio for debriefing time \\
$\delta$ & Social cueing efficiency of the avatar during debriefing \\
$\gamma = c \cdot \delta$ & Overall debriefing efficiency factor ($z_v = \gamma y_v$) \\
\bottomrule
\end{tabular}
  \vspace{-20pt}
\end{table}

\section{Optimization of the user's spare time:\\An NP-hard problem} \label{sec:prob}

Here, we define the problem of maximizing the user's spare time within a given time period. 
We prove that the problem is computationally intractable, so we can skip its formal formulation (e.g., as a program), and focus directly on the algorithmic design, either through heuristics or through approximation algorithms. 

\begin{problem} \label{pro::main}
Given (i) a user $u$, (ii) her independent avatar $a$, (iii) a set of $n$ alters and their respective initial non-avatar-mediated time partitioning $\tilde{x_1} + ... + \tilde{x_n} = \tilde{X}$, (iv) the time slots $d_i$, (v) a conflict graph $G$, and (vi) the social requests $R_v$: Allocate the social requests in the available time slots, and maximize the spare time of $u$ (i.e., $\tilde{X} - X - Z$) by also using $a$'s avatar-mediated time $Y$ for socialization in a resulting total combined time partitioning $(x_1 + y_1) + ... + (x_n + y_n) = X + Y$ with $Y \geq 0$. 
\end{problem}
We define the metaverse user spare time maximization decision problem as follows.

\begin{definition}[MUST-M]
The metaverse user spare time maximization (MUST-M) is the problem of deciding whether, given Problem \ref{pro::main} and a number $S \geq 0$, the user can spare at least $S$ time after using the avatar time. 
\end{definition}

We shall now prove that this problem is NP-hard by reducing the CLIQUE problem \cite{garey1979computers} to a special case of MUST-M. CLIQUE has been shown to be NP-complete as one of Karp's 21 NP-complete problems \cite{Karp1972}, and is as follows.

\begin{definition}[CLIQUE] \label{def::cli}
Given a graph $G_c=(V_c,E_c)$, and positive integer $K_c \leq |V_c|$, does $G_c$ contain a clique of sized $K_c$ or more? (i.e., a subset $V'_c \subseteq V_c$ with $|V'_c| \geq K_c$ such that every two vertices in $|V'_c|$ are joined by an edge in $E_c$)
\end{definition}

\begin{theorem}
MUST-M is NP-hard 
\end{theorem}

\begin{proof}
Consider the CLIQUE problem of Definition \ref{def::cli}. We will now create an instance of MUST-M from CLIQUE by performing the following actions: We set the available avatar-mediated time as $Y = 0$, the number of time slots as $k = K_c - 1$, the time slot size as $s(d_i) = |V_c|$ for all time slots, the size of all lists as $|R_s| = 1$ the size of all the requests\footnote{Here, by the value of $1$, we mean the smallest time value that cannot be fragmented any more; more fragmentation would not make sense from the socialization utility standpoint. In the formulation, the value $1$ could be replaced by some constant.} as $r^1_v = 1$, the deadlines as $d'^1_v = d_1$ and $d''^1_v = d_k$ for all deadlines, the conflict graph of the alters as $G = G_c$, and the non-avatar-mediated times as $\tilde{X} = \sum_v r^1_v$. Finally, we set $S = 0$, and therefore, any feasible solution gives a positive answer to MUST-M. Notice that a solution to this instance of MUST-M would provide an answer to CLIQUE, which means that CLIQUE $\leq_m$ MUST-M. This completes the proof.
\end{proof}

\section{An approximation of the optimal spare time} \label{sec::algo}

As shown in section \ref{pro::main}, MUST-M problem is proven to be NP-hard. Therefore, it can not be solved optimally in polynomial time. To tackle this, we design a heuristic algorithm according to the following strategy: In the first step, we simplify the problem by temporarily removing some of its constraints. Specifically, we assume the simplified version of MUST-M where there are no conflicts, there are no time slots, and there are infinite deadlines: The objective now is to just come up with the values of $X$ and $Y$. This makes the simplified version easier to handle, allowing us to find an optimal solution for it. Once we have this solution, we reintroduce the removed constraints and continue solving the rest of the problem using a spare-time-centric approach. Throughout this process, we make sure to build on the optimal solution from the first step, improving the overall quality of our final result.




\subsection{The optimal solution to the simplified version of MUST-M}

We assume that there are no conflicts, there are no time slots, and there are infinite deadlines. Following the problem definition and the discussion in the previous sections, in order for $u$'s 
time to be augmented, $u$ could delegate to $a$ a portion of the socializing activity time. In order to do so, $u$ would have to exchange some of the actual non-avatar-mediated time for implementing $a$'s debriefing process, with a final target to achieve an augmented social presence with 
$a$'s avatar-mediated time. In this case, $u$ would be able to have an ego network comprising the alters in $V$ 
in a resulting amount of time across $u$ and $a$ combined. 
%
%
%
The problem can now be framed as in Opt. Problem~\ref{prob:simpl_must-m}.

\begin{Problem}
\begin{align}
  \min\quad & \sum_v x_v + \gamma \cdot y_v  \label{opt:obj_fun}\\\
\text{s.t.\quad} & \nonumber \\
& x_v + \frac{1}{\beta_v} \cdot y_v \geq \tilde{x}_v \label{opt:c1} &   & \forall v  \\
& \sum_v x_v + \gamma \cdot y_v \leq \tilde{X}' \label{opt:c2}&   & \\
& \sum_v y_v \leq Y \label{opt:c3}&   & \\
& \sum_v \gamma \cdot y_v \leq Z \label{opt:c4}&   & \\ 
& x_v, y_v \geq 0 \label{opt:c5} &   & \forall v
\end{align}
\caption{Simplified version of MUST-M.}
\label{prob:simpl_must-m}
\end{Problem}

The objective function (Eq.~\ref{opt:obj_fun}) minimizes the sum of the non-avatar-mediated time and the debriefing time, thus maximizing the available non-avatar-mediated spare time on the user side. Constraints in Eq.~\ref{opt:c1} guarantee that the resulting effective non-avatar-mediated and avatar-mediated socializing time for each $v$ is not less than the initial non-avatar-mediated time, i.e., that the resulting relationships are not modified in terms of cognitive involvement with respect to the case when the avatar is not present. Constraint in Eq.~\ref{opt:c2} guarantees that the non-avatar-mediated time and the debriefing time do not surpass the total available time of the user. Constraint in Eq.~\ref{opt:c3} guarantees that the avatar-mediated time does not surpass the total available time of the avatar. Constraint \ref{opt:c4} guarantees that the debriefing process abides by the non-distractibility rule introduced with Eq.~\ref{eq:zmax}. Finally, constraints in Eq.~\ref{opt:c4} guarantee that time values are zero or positive.

\subsection{Feasibility}

\newcommand{\cc}[1]{\textcircled{\raisebox{-0.9pt}{#1}}}

The problem we have formulated in the previous section can be solved using any standard LP solver. For the purpose of our investigation, it is interesting to analyse the feasible region of the problem, where a solution can be obtained. We assume here that $\beta$ and $\gamma$ are constant and equal for all alters. 

Due to the minimization nature of the problem, the first constraint (Eq.~\ref{opt:c1}) can only hold as an equality. Consequently, the constraint can be used to perform a variable substitution in the objective function, as $x_v$ can be expressed as a function of $y_v$ and the initial physical time allocated to that alter $\tilde{x}_v$. Specifically, it holds the following:
\begin{equation}
    x_v = \tilde{x}_v - \frac{1}{\beta} y_v.
    \label{eq:xv_sub}
\end{equation}
Effectively, we just need to find a solution for $y_v$ and the values for $x_v$ will follow. This allows us to rewrite the optimization problem as shown in Opt. Problem~\ref{prob:feasibility}. 

\begin{Problem}
\begin{align} 
  \min\quad & \tilde{X} + \left(\gamma - \frac{1}{\beta}\right) \sum_v y_v \label{opt2:obj_fun} \\
\text{s.t.\quad} & \nonumber \\
& \left(\gamma - \frac{1}{\beta}\right) \sum_v y_v \leq (\tilde{X}' - \tilde{X}) \label{opt2:c2}&   & \\
& \sum_v y_v \leq \min \left\{ Y, \frac{Z}{\gamma} \right\} \label{opt2:c3}&   & \\
& 0 \leq y_v \leq \beta \tilde{x}_v \label{opt2:c1} &   & \forall v
\end{align}
\caption{Rewriting Opt. Problem~\ref{prob:simpl_must-m}.}
\label{prob:feasibility}
\end{Problem}


Let us now look at the first constraint (Eq.~\ref{opt2:c2}), which effectively requires that the debriefing time is always smaller than the time user $u$ should have spent directly interacting with alter $v$ (if this were not the case, user $u$ would be better off performing non-avatar-mediated interactions, as avatars would bring no benefit). Since, according to the constraint in Eq.~\ref{opt2:c1}, $y_v$ are always positive, the first constraint can only be satisfied if $\gamma \leq \frac{1}{\beta}$ (case \cc{A}) or if $y_v = 0$ for all $v$ (case \cc{B}). The discussion in Section~\ref{sec:model} tells us that in a realistic setting $\frac{1}{\beta} \sim 0.78$ and $\gamma \sim 0.63$, so we are operating in case \cc{A} and non-trivial solutions (i.e., $y_v > 0$) are possible.
The remaining constraint (Eq.~\ref{opt2:c3}) is a simple cap on the total avatar time. Its concrete meaning is that the total time the avatar spends socializing cannot be greater than the overall time in the avatar's capacity for socialization or than the maximum debriefing capacity of the user.

\begin{figure}
    \centering
    \includegraphics[width=.7\linewidth]{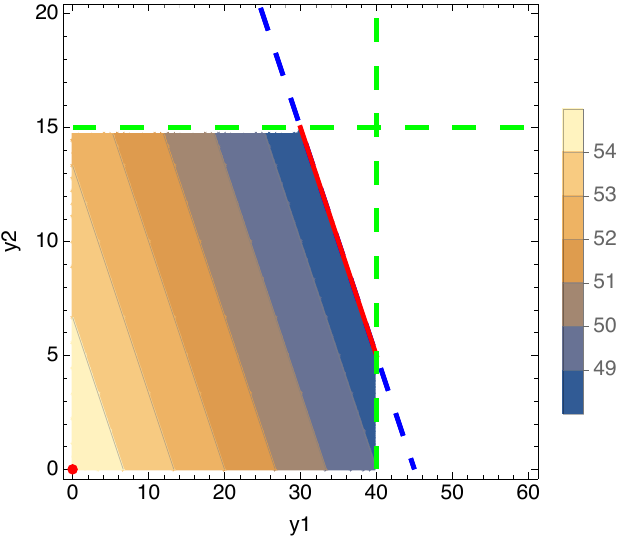} 
    \caption{Contour plot of the objective function in Eq.~\ref{opt2:obj_fun} when the socialization involves two alters. The dashed blue line corresponds to the constraint in Eq.~\ref{opt2:c3}, the dashed green lines to those in Eq.~\ref{opt2:c1}. The red line is the minimizing solution for case \cc{A}, the red point for case \cc{B}. The plot is obtained with $\tilde{x}_1 = 40, \tilde{x}_2 = 15, \gamma = 0.63, \frac{1}{\beta} = 0.78, Y = 45, Z = 45$.} \vspace{-10pt}
    \label{fig:feasible-region}
\end{figure}

In light of the reformulated problem, we can now revisit the objective function $\min \tilde{X} + (\gamma  - \frac{1}{\beta}) \sum_v y_v$. Let us distinguish the two cases, \cc{A} and \cc{B}, that we have already introduced. When $(\gamma  - \frac{1}{\beta}) > 0$ \cc{B}, the objective function is minimized when $y_v = 0$ for all $v$. Vice versa, when $(\gamma  - \frac{1}{\beta}) \leq 0$ \cc{A}, as in the realistic conditions discussed in Section~\ref{sec:model}, the objective function is minimized when $y_v$ take their highest values within the feasible region. As illustrated in Figure~\ref{fig:feasible-region}, in case \cc{A}, the objective function is minimized when $y_v$ lies on the red solid line, while in case \cc{B} when $y_v = 0$. 

The above discussion allows us to identify two main regimes in our avatar-mediated socialization problem. On the one hand, when the debriefing process is less efficient than the avatar socialization process (i.e., when $\gamma > \frac{1}{\beta}$), delegating socialization to the avatar is not convenient because the time saved from socialization is spent in debriefing. The optimal solution in this case is $y_v = 0$, which is equivalent to not using the avatar at all. Vice versa, when the debriefing process is more efficient than the avatar socialization process (i.e., when $\gamma \leq \frac{1}{\beta}$), the optimal solution consists in using the avatar as much as possible, or at least as much as the constraints allow. In short, with these linear constraints and objective function, depending on the socialization and debriefing efficiencies, the user should either not use the avatar at all or delegate as much as possible to it. This conclusion aligns intuitively with our problem formulation.

\subsection{Insights from the formulation}
\label{sec:results}

We wrapped the previous section with a qualitative finding: when the debriefing process is more efficient than the avatar socialization process, using the avatar is always convenient. In this section, we delve into a quantitative analysis of this benefit, focusing on the non-pathological case $X = \tilde{X}$ (i.e., we assume that the user is already able to sustain its ego network and we assess the spare time that can be gained by using the avatar).

To this aim, we leverage the model provided in~\cite{Conti2012} for generating realistic ego networks. We fix $\beta$ to $1.29$ in our optimization problem and we vary $\gamma$ between $[0,\frac{1}{\beta})$ (i.e., we remain in the range where case \cc{A} holds but we assume, as it is realistic, that there will be some variability in the $\gamma$ values). We test scenarios where the total avatar time $Y$ is greater, equal, or smaller than the total user socialization time without the avatar ($\tilde{X}$). In Figure~\ref{fig:spare-time}, we plot the amount of spare time gained by the user by leveraging the avatar. When $\gamma$ and $\frac{1}{\beta}$ are close (i.e., their difference is near zero), the gained spare time is very small. However, as the debriefing becomes more and more efficient with respect to the avatar socialization, then the time saved by the user becomes increasingly higher. 

To quantify the potential impact of the spare time, we compare the saved time against the time an average user spends on individual alters of their ego network layers. We extract these numbers from~\cite{Conti2012}, where statistics for three concentric layers are reported. The layers are referred to as the ``support clique'', ``sympathy group'', and ``active network'', with average cumulative sizes of 4.6, 14.3, and 132.5 alters, respectively. This structure deviates from the standard ego network structure due to the exclusion of the ``affinity group" layer, whose properties are not well-defined in the related literature.
From~\cite{Conti2012}, we obtain $t_{\textrm{active}}=8.81 h$, $t_{\textrm{sympathy}} = 38.72 h$, $t_{\textrm{support}}=74 h$.
Considering together the time spent on all the alters in each layer, we obtain $T_{\textrm{active}}=1041.51 h$, $T_{\textrm{sympathy}} = 375.83 h$, $T_{\textrm{support}}=340.65 h$. 
%

\begin{figure}
    \centering
    \includegraphics[width=.8\linewidth]{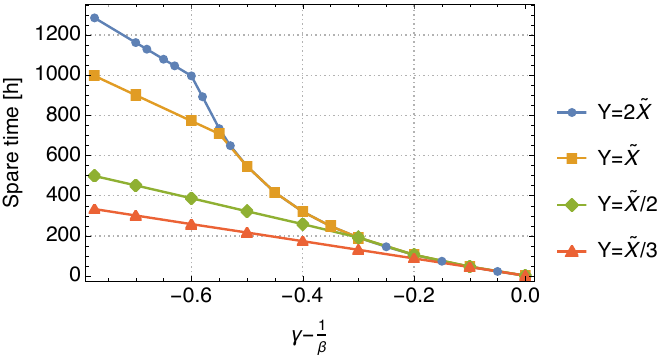}
    \caption{Spare time for an ego network with $|V| = 117, \tilde{X} = 1288h, Z = 300h, \beta = 1.29$.}  \vspace{-18pt}
    \label{fig:spare-time}
\end{figure}

It is evident that when the amount of time saved is high enough, users might have the opportunity to reinvest it in expanding their ego network circles and enhancing their social connections. For instance, with approximately 1000 spare hours, one could potentially double the active layer of their ego network, or with around 350 spare hours, strengthen their support clique. While fully characterizing the dynamics of this augmentation is beyond the scope of this work, and the relationship between time saved and new connections gained might not be as direct and straightforward, these initial insights suggest that avatars in the Metaverse have the potential to increase a user's available time -- a resource that could be redirected toward growing their personal social network. In theory, this could result in surpassing Dunbar's number of 150 meaningful relationships per ego, thus fostering a significantly enriched social life within the Metaverse -- an advancement not observed in online relationships within the social web~\cite{DUNBAR201539}.

\subsection{Reintroducing the original constraints}

Given the $X$ and $Y$ output values of the previous step, we design a formulation and a heuristic that distributes the time across time slots, trying to satisfy the social requests. The heuristic formulation is presented in Opt. Problem~\ref{prob:heuristic}, which uses the following social cost function\footnote{Drawing on the framework of social norms of \cite{cialdini_social_1998}, we model the social cost of delay using a threshold function. This reflects the idea that social groups often have implicit expectations about acceptable response times for social requests. Delays falling within this expected window incur minimal cost. However, exceeding a certain threshold, which represents a clear violation of descriptive (typical behavior) or injunctive (socially approved behavior) norms of responsiveness and helpfulness, leads to a significant increase in social cost due to social disapproval and potential damage to relationships. Also, drawing upon the foundational argument of \cite{d7019d62-38e8-3496-94e6-216994c35897} that the norm of reciprocity creates a moral obligation to return benefits, we model the social cost of delay as increasing linearly with time. This reflects the idea that the longer a social request remains unfulfilled, the greater the unmet obligation and the more prolonged the disruption to the expected social exchange, potentially leading to a steady increase in negative social perception and a weakening of reciprocal ties.}\footnote{Although the scheduling model is originally defined over a one-year time horizon, we extend the temporal scope to a two-year back-to-back setup when computing the social cost function. This extension is introduced exclusively to ensure fairness in handling social interaction requests that cannot be satisfied within the first-year allocation due to slot reservations. Specifically, each reserved slot allocated during the first year is virtually duplicated in the corresponding position of a second, identical year. Social interaction requests are then also evaluated with respect to their feasibility in this second-year instance. If a request cannot be accommodated within the first year, but can be matched in the second year under the duplicated reservation, the cost is still computed, albeit shifted temporally, by assigning it a time penalty equal to its corresponding time slot in the first year plus 365 days. This modeling trick ensures that the scheduling mechanism accounts for future compatibility and avoids penalizing agents solely due to year-bound limitations, thereby promoting long-term fairness across interactions without modifying the original one-year model structure.}:

\begin{equation}
f_c(v,i,j) = \left\{
\begin{aligned}
i + 365 - d''^j_v & : i < d'^j_v\\
 0 & : d'^j_v \leq i \leq d''^j_v\\
 i - d''^j_v & : i > d''^j_v
\end{aligned} \right.
\end{equation}

\begin{Problem}[t!]
\begin{align}
  \min\quad & \sum_v \sum_i \sum_j f_c(v,i,j) \cdot (W^x_{vij} + W^y_{vij})  \label{opt:obj_fun2}\\\
\text{s.t.\quad} & \nonumber \\
& \sum_v \sum_i \sum_j r^j_v (W^x_{vij} + W^y_{vij}) = X+Y \label{opt:c23}&   & \\
& \sum_v \sum_j r^j_v (W^x_{vij} + \gamma \cdot W^y_{vij}) \leq s(d_i) \label{opt:c24}&   & \forall i \\
& \sum_v \sum_j r^j_v \cdot W^y_{vij} \leq s(d_i) \label{opt:c25}&   & \forall i \\
& \sum_i \sum_j r^j_v \cdot W^x_{vij} = x_v \label{opt:c251}&   & \forall v \\
& \sum_i \sum_j r^j_v \cdot W^y_{vij} = y_v \label{opt:c252}&   & \forall v \\
& (W^x_{vij} + W^y_{vij}) \leq 1  & & \forall j  \label{opt:c21}\\
& (W^x_{v_1ij} + W^x_{v_2ij}) \leq 1 \label{opt:c22} &   & \forall (v_1, v_2) \in G, \forall i\\
& W^x_{vij}, W^y_{vij} \in \{0,1\} \label{opt:c27} &   &
\end{align}
\caption{Heuristic formulation}
\label{prob:heuristic}
\end{Problem}

Variables $W^x_{vij}$ and $W^y_{vij}$ can be set to either $0$ or $1$ according to constraints \ref{opt:c27}, and their value signifies whether the request $j$ of alter $v$ is served on day $i$, with either non-avatar-mediated time or avatar-mediated time respectively.  $W^x_{vij}$ and $W^y_{vij}$ can not be set to $1$ at the same time, due to constraints \ref{opt:c21}. 

The objective function \ref{opt:obj_fun2} minimizes the cost function $f_c$ value across all socialization requests. Constraints \ref{opt:c23} force the size of the socialization requests time to be equal to the sum of the available global optimized non-avatar-mediated and avatar-mediated time. Constraints \ref{opt:c24} force the sum of the daily socialization requests served by non-avatar-mediated interactions and debriefing time size to be equal to or less than the available physical time of the user. In a similar manner, but from the side of the avatar, constraints \ref{opt:c25} force the daily socialization requests served by non-avatar-mediated interactions' time size to be equal to or less than the available socialization time of the avatar. Constraints \ref{opt:c251} and \ref{opt:c252} force the non-avatar-mediated and the avatar-mediated time, respectively, for each alter to be equal to the equivalent individual optimal times. Finally, constraints \ref{opt:c22} impose the physical constraints of the conflict graph for every pair of alters. 


We design a spare time-centric heuristic that schedules social interactions by prioritizing non-avatar-mediated encounters, while ensuring that user time constraints and debriefing requirements are met. The process proceeds in five main steps, as detailed in Fig.~\ref{alg:adaptive_heuristic}. First, the heuristic computes the optimal yearly time allocation for each alter, balancing non-avatar- and avatar-mediated interactions by minimizing a weighted sum of their respective durations. Second, it estimates the social cost of each request across feasible days, based on factors such as proximity to deadlines and time availability. In the third step, requests are sorted by increasing cost for each day to prioritize low-cost interactions. The fourth step attempts to assign each request to a non-avatar-mediated slot, ensuring that the user's daily availability and conflict avoidance constraints are respected. Any remaining unscheduled requests are considered for avatar mediation in the fifth step, provided that sufficient avatar availability and user debriefing capacity exist. The heuristic returns a valid schedule only if all constraints, including total time budgets and debriefing limits, are satisfied. This approach favors early allocation of high-value interactions during the user's free time while gracefully degrading to avatar-mediated scheduling when necessary.

\begin{figure}
\caption{Spare time-centric with sorted sequential request scheduling}
\label{alg:adaptive_heuristic}

\SetAlgoLined
\SetKwInOut{Input}{Input}
\SetKwInOut{Output}{Output}

\Input{User socialization capacity $\tilde{X}'$, debriefing threshold $Z$, non-avatar-mediated time $Y$, list of requests $(v, j)$, daily timeslots size $s(d_i)$}
\Output{Non-avatar-mediated and avatar-mediated interactions schedule $W^x_{vij}, W^y_{vij}$}

\BlankLine
\textbf{Initialization:}\;
Initialize schedules $W^x_{vij} = W^y_{vij} = 0$, track available times, and create a list of requests\;


\BlankLine
\textbf{Step 1: Optimize the yearly times:}\;

Assign values to all $x_v$ and $y_v$, by optimally solving the program \ref{opt:obj_fun} ($\min  \sum_v x_v + \gamma \cdot y_v$)\;

\BlankLine
\textbf{Step 2: Calculate costs of each request for each day:}\;

\ForEach{request $(v, j)$}{
  \ForEach{day $i$ 
  }{
    Calculate the cost $f_c(v, i, j)$\;
  }
}

\BlankLine
\textbf{Step 3: Sort requests by minimum cost for each day:}\;

\ForEach{day $i$}{
Sort the list of requests by their minimum costs in ascending order\;
}

\BlankLine
\textbf{Step 4: Attempt non-avatar-mediated allocation, avoiding conflicts:}\;

\ForEach{day i}{
\ForEach{request $(v, j)$}{
  \If{there is no conflict of v in i}{
  Attempt to allocate it as non-avatar-mediated on its minimum cost day\;
  
  \If{allocation successful}{

    Set $W^x_{vij} = 1$\;
    
    
    Update available daily user time \;
    
    Mark the request as scheduled\;
  }
}
}
}

\BlankLine
\textbf{Step 5: Attempt avatar-mediated allocation:}\;

\ForEach{day $i$}{
\ForEach{remaining unscheduled request}{
    \If{there is available time for debriefing in $i$}{
  Attempt to allocate it as avatar-mediated on its minimum cost day\;
  
  \If{allocation successful}{
  
    Set $W^y_{vij} = 1$\;
    
    
    Update available daily avatar time \; 
    
    
    Update available daily user time (debriefing costs) \;
    
    Mark the request as scheduled\;
  }
}
}
}



\BlankLine
 \If{the constraints of program \ref{opt:obj_fun2} are met}{
\textbf{Output the schedule:}\;
Return the $W^x_{vij}$ and $W^y_{vij}$ values\;
}

\end{figure}

\section{Performance evaluation} \label{sec::eval}


In this section, we present the simulation experiments conducted to assess the impact of the introduction of an independent avatar to social setting. We first present the experimental setup and parameter setting, and then we present and discuss the simulation results. 

\subsection{Experimental parameter setup}

To perform a representative and insightful evaluation, we first generated 10,000 ego networks using a custom Python simulator that implements the ego network generation process described in~\cite{Conti2012}. Each ego network consists of a central user and a variable number of alters. To ensure coverage of diverse configurations, we selected three representative ego networks corresponding to the 10th, 50th, and 90th percentiles in terms of size. For the set of users in Fig.~\ref{fig:histogram}, these percentiles correspond to ego networks containing 68, 126, and 170 alters, respectively.

\begin{figure}[t!]
    \centering
    \includegraphics[width=0.7\columnwidth]{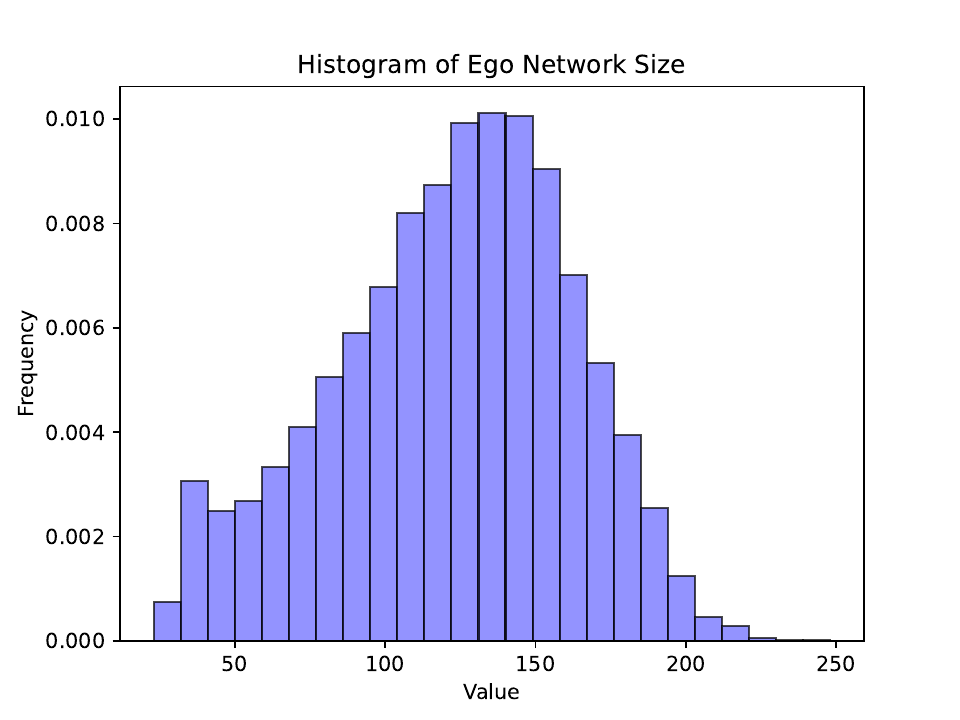}
    \caption{The histogram of the size of the ego networks generated by the simulator.}
    \label{fig:histogram}
      \vspace{-18pt}
\end{figure}

\begin{figure*}[t!]
    \centering
    \subfloat[$n = 68$.]{%
        \includegraphics[width=0.30\textwidth]{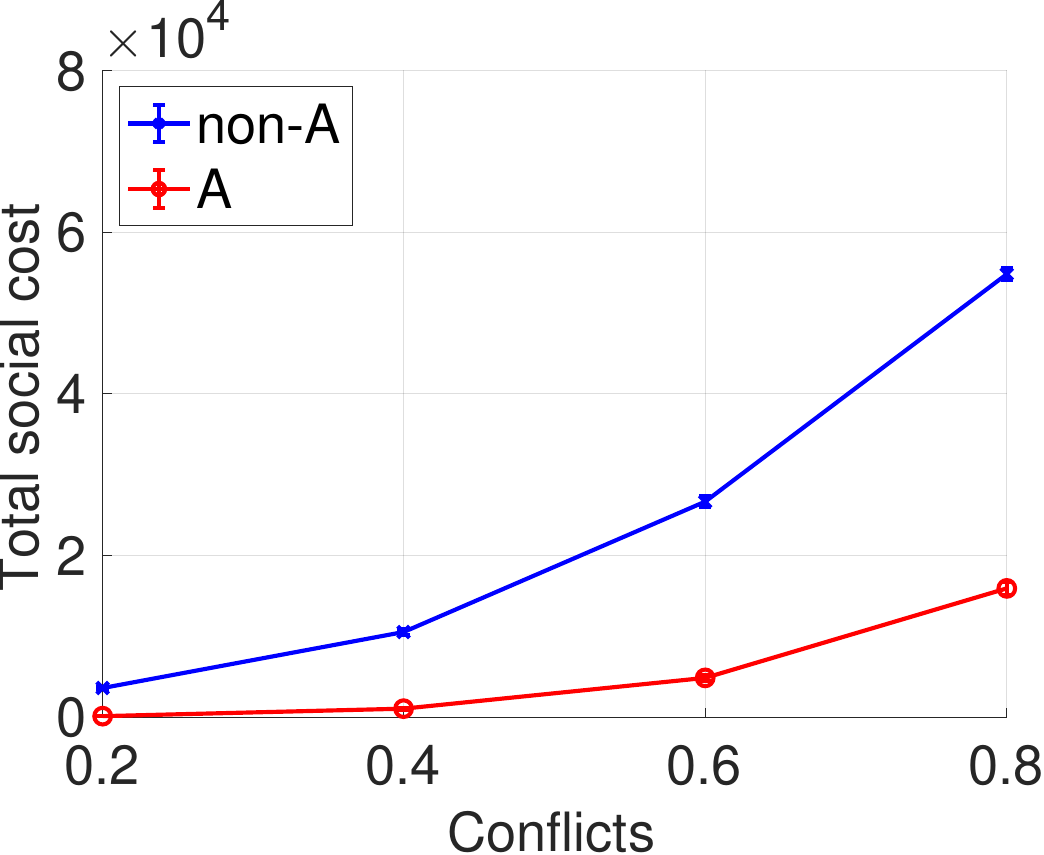}
        \label{fig:varyConflicts-sub1}
    } \hspace{2mm}  
    \subfloat[$n = 126$.]{%
        \includegraphics[width=0.30\textwidth]{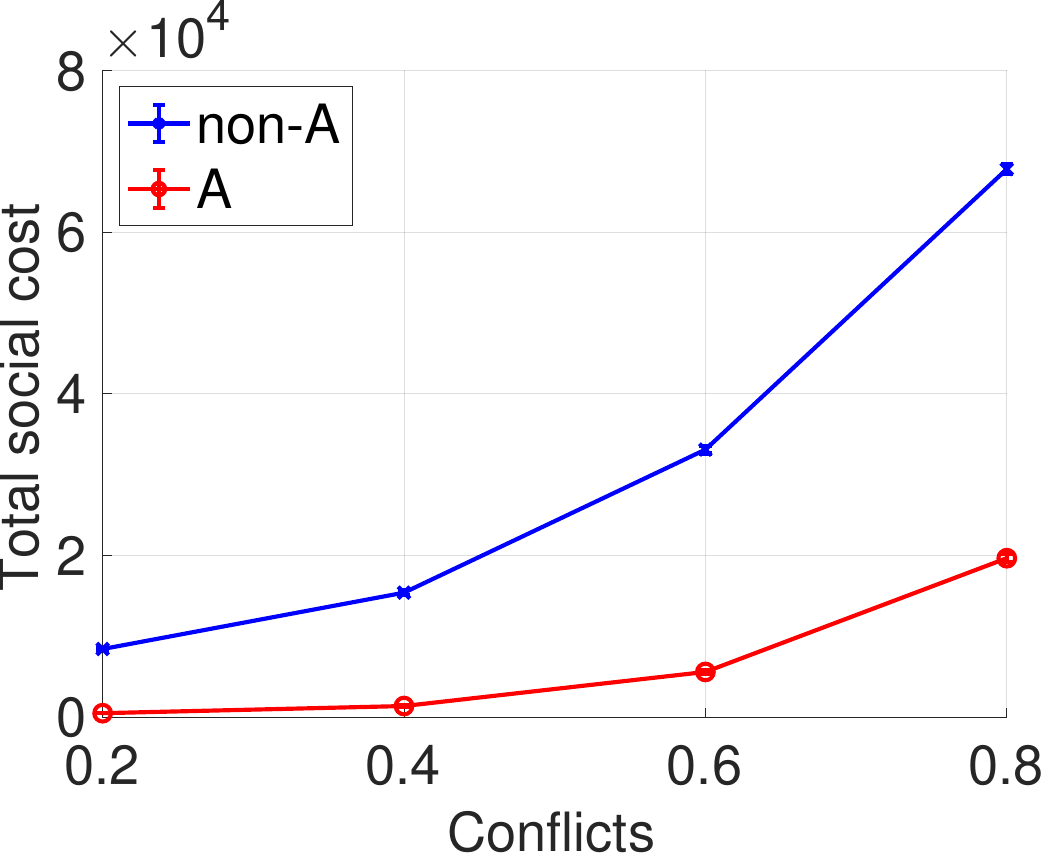}
        \label{fig:varyConflicts-sub2}
    } \hspace{2mm}  
    \subfloat[$n = 170$.]{%
        \includegraphics[width=0.30\textwidth]{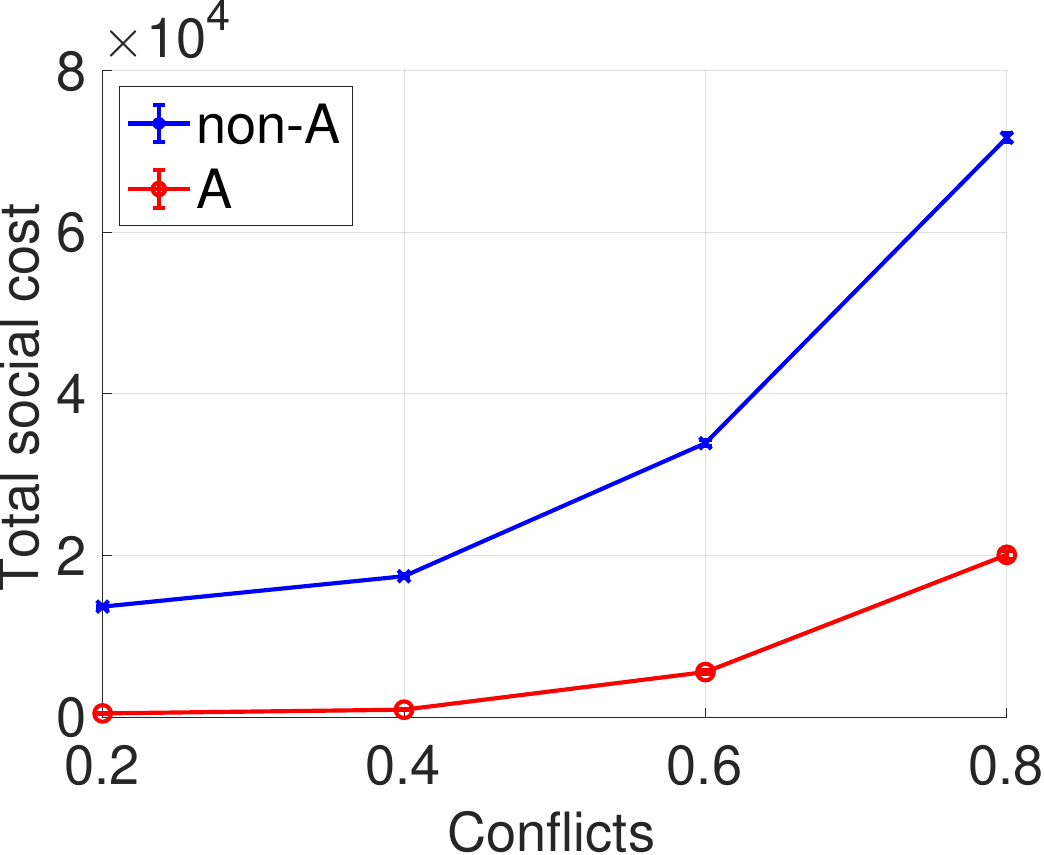}
        \label{fig:varyConflicts-sub3}
    } 

    \centering
    \subfloat[$n = 68$.]{%
        \includegraphics[width=0.30\textwidth]{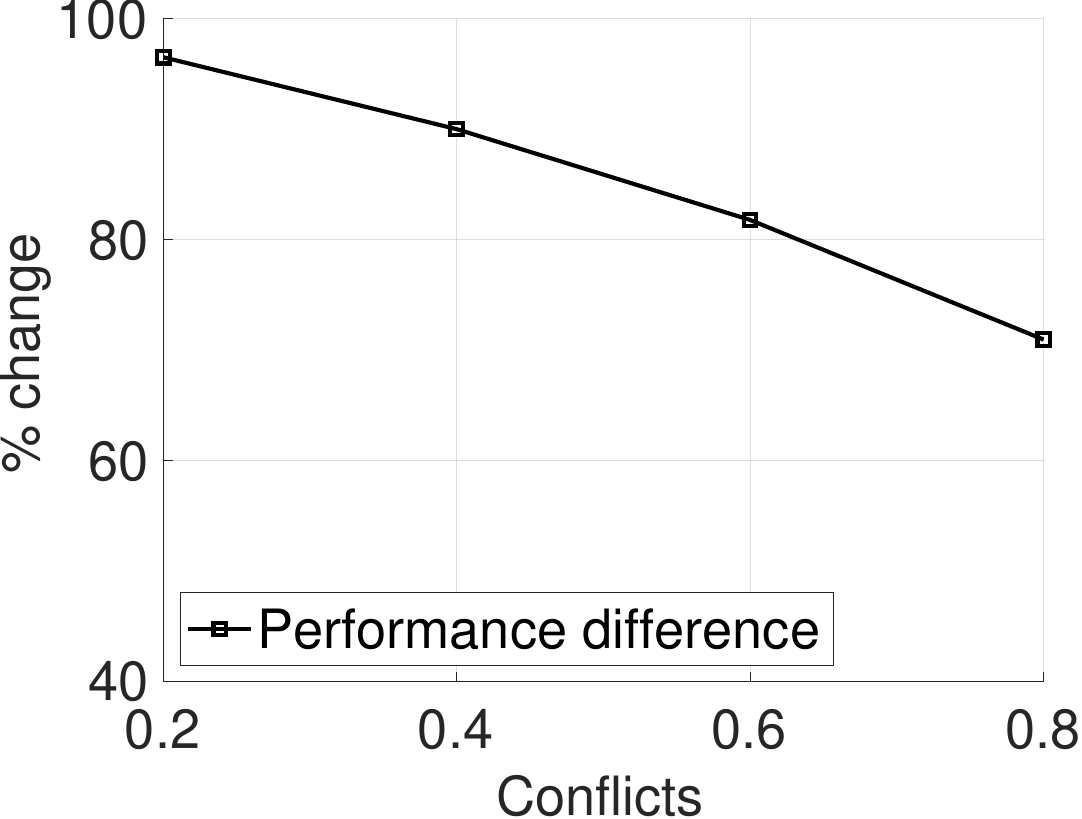}
        \label{fig:sub1}
    } \hspace{2mm}  
    \subfloat[$n = 126$.]{%
        \includegraphics[width=0.30\textwidth]{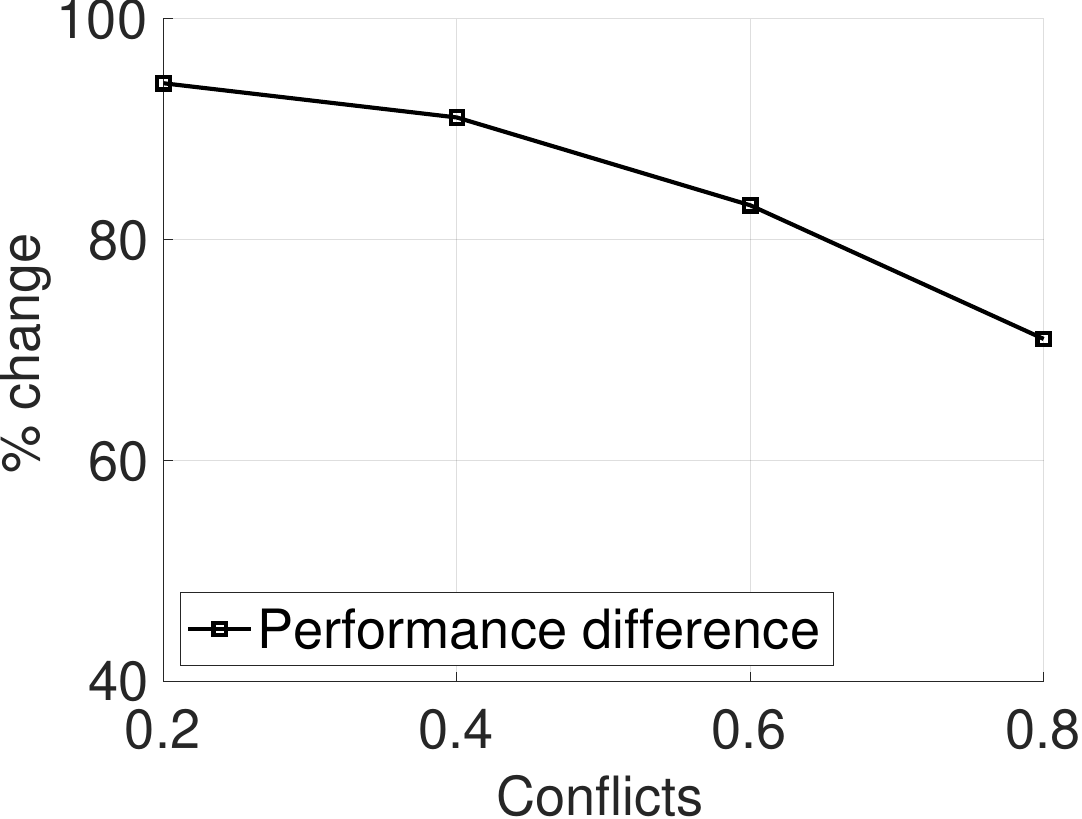}
        \label{fig:sub2}
    } \hspace{2mm}  
    \subfloat[$n = 170$.]{%
        \includegraphics[width=0.30\textwidth]{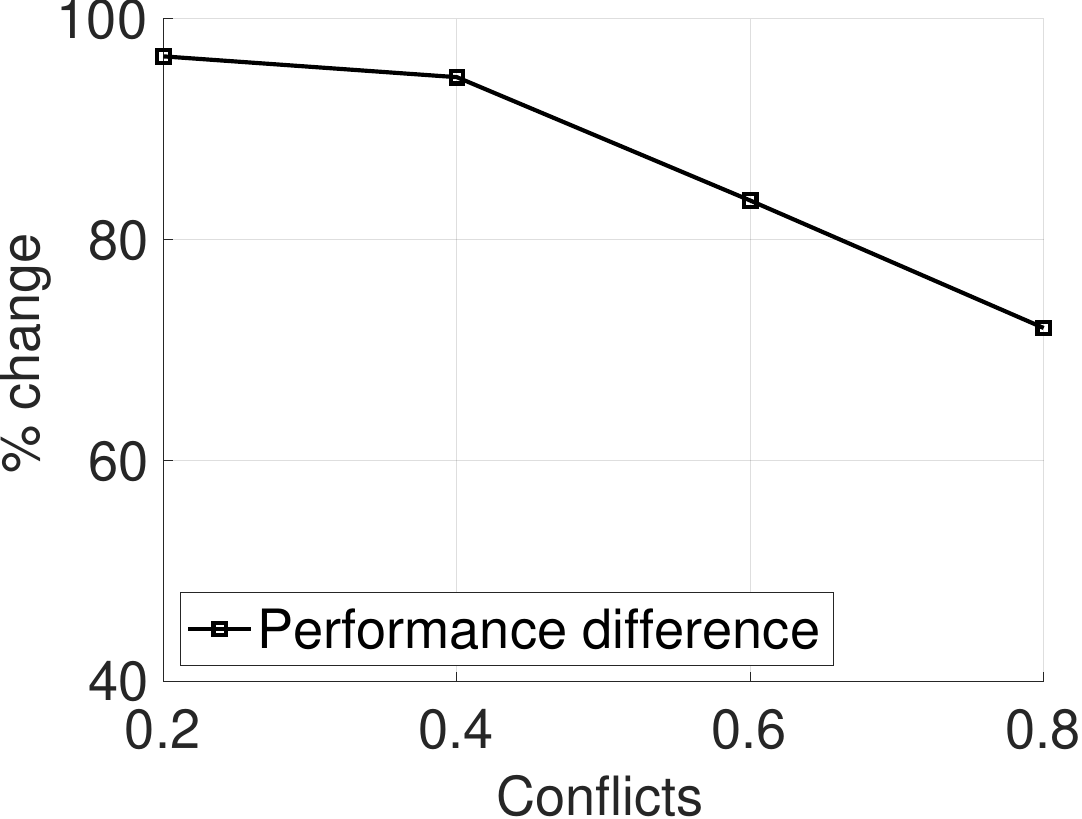}
        \label{fig:sub3}
    } 
    \caption{Total social cost and performance difference, for varying number of conflicts, with $|E| = \frac{n\cdot(n-1)}{2} \cdot \{ 0, 0.2, 0.4, 0.6, 0.8 \}$.}
    \label{fig:varyConflicts}
\end{figure*}

\begin{figure*}[t!]
    \centering
    \subfloat[$n = 68$.]{%
        \includegraphics[width=0.30\textwidth]{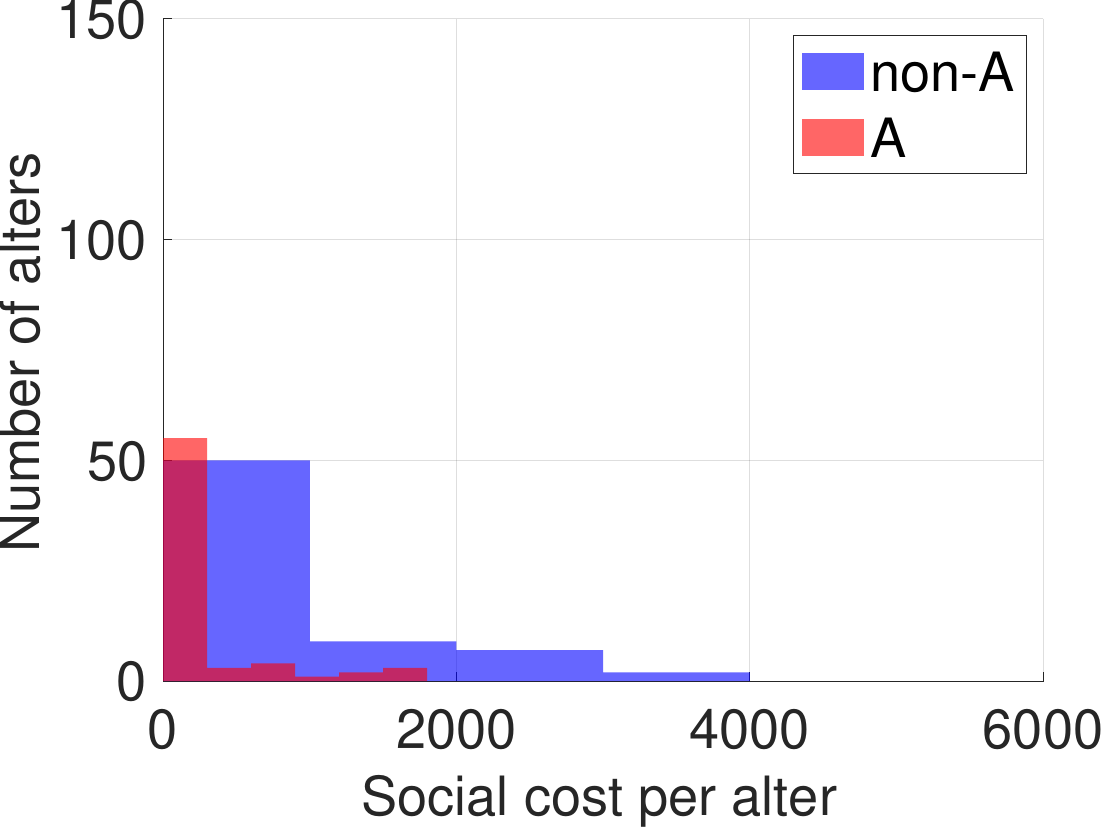}
        \label{fig:varyConflicts-histograms:sub1}
    } \hspace{2mm}  
    \subfloat[$n = 126$.]{%
        \includegraphics[width=0.30\textwidth]{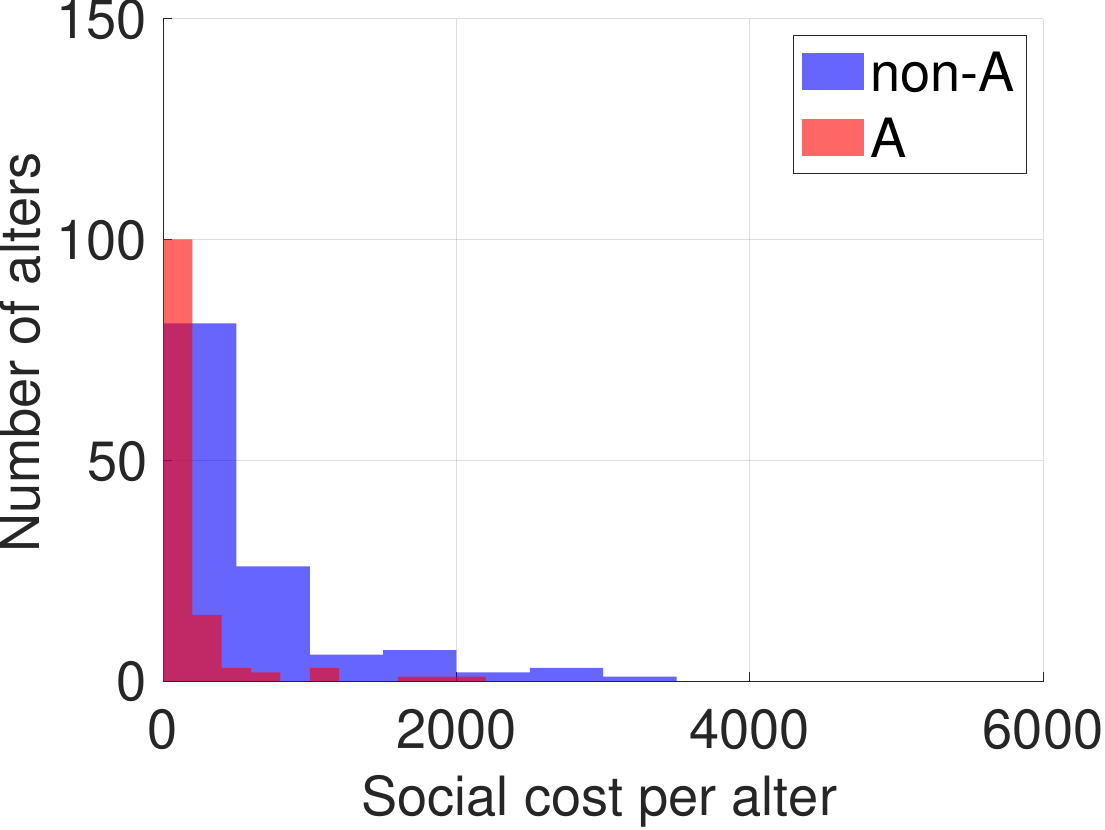}
        \label{fig:varyConflicts-histograms:sub2}
    } \hspace{2mm}  
    \subfloat[$n = 170$.]{%
        \includegraphics[width=0.30\textwidth]{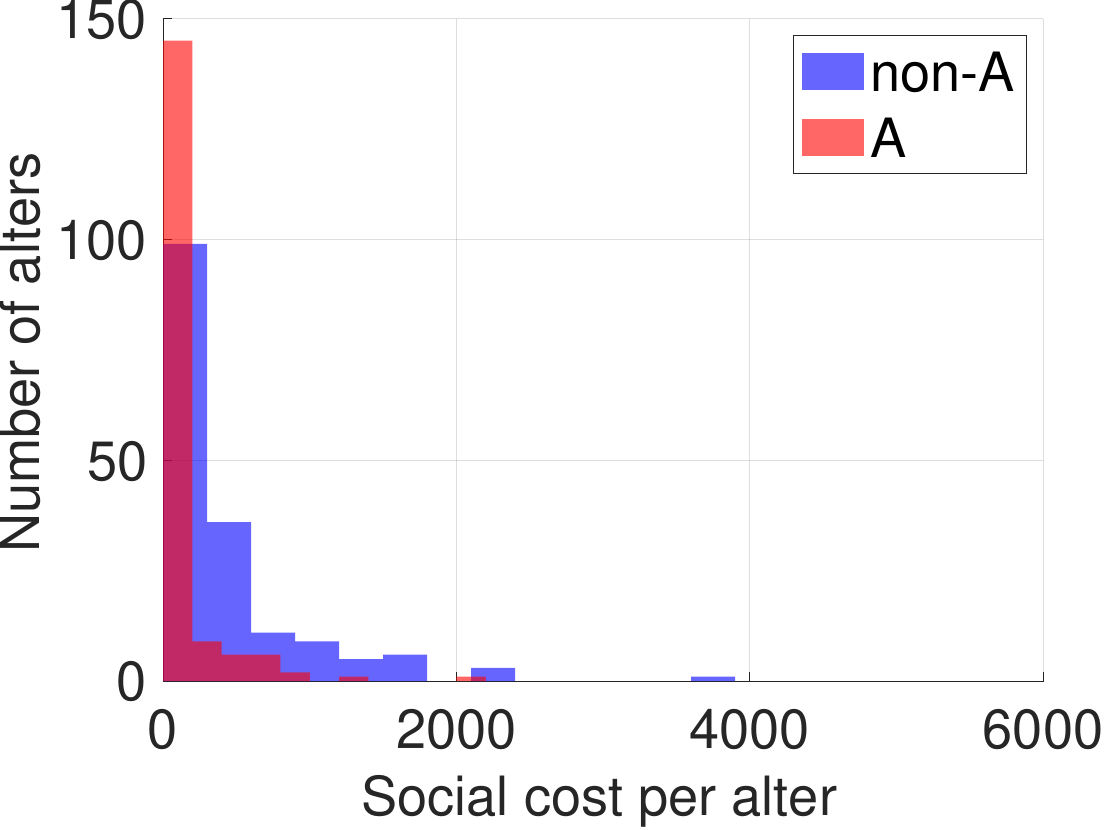}
        \label{fig:varyConflicts-histograms:sub3}
    } 
    \caption{Social cost per alter distribution.}
    \label{fig:varyConflicts-histograms}
\end{figure*}

\begin{figure*}[t!]
    \centering
    \subfloat[Varying deadlines.]{%
        \includegraphics[width=0.3\textwidth]{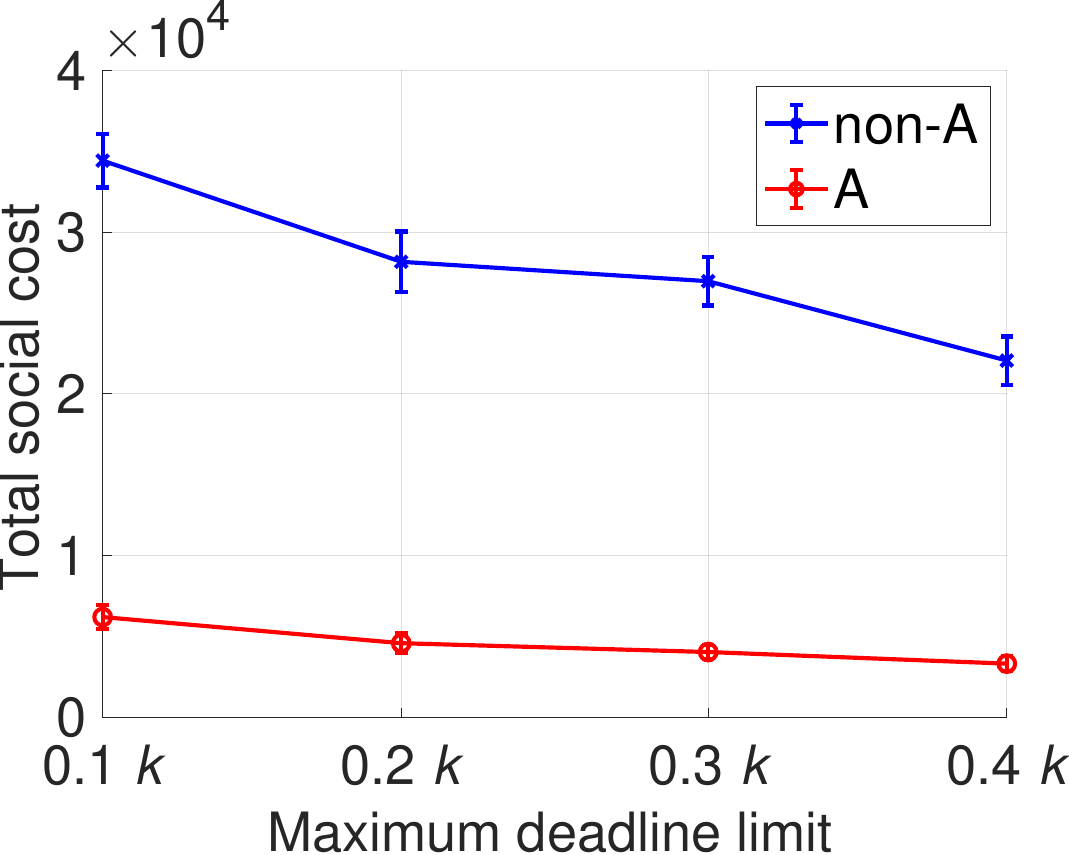}
        \label{fig:metrics-sub1}
    } \hspace{2mm}  
    \subfloat[Varying $Y$.]{%
        \includegraphics[width=0.3\textwidth]{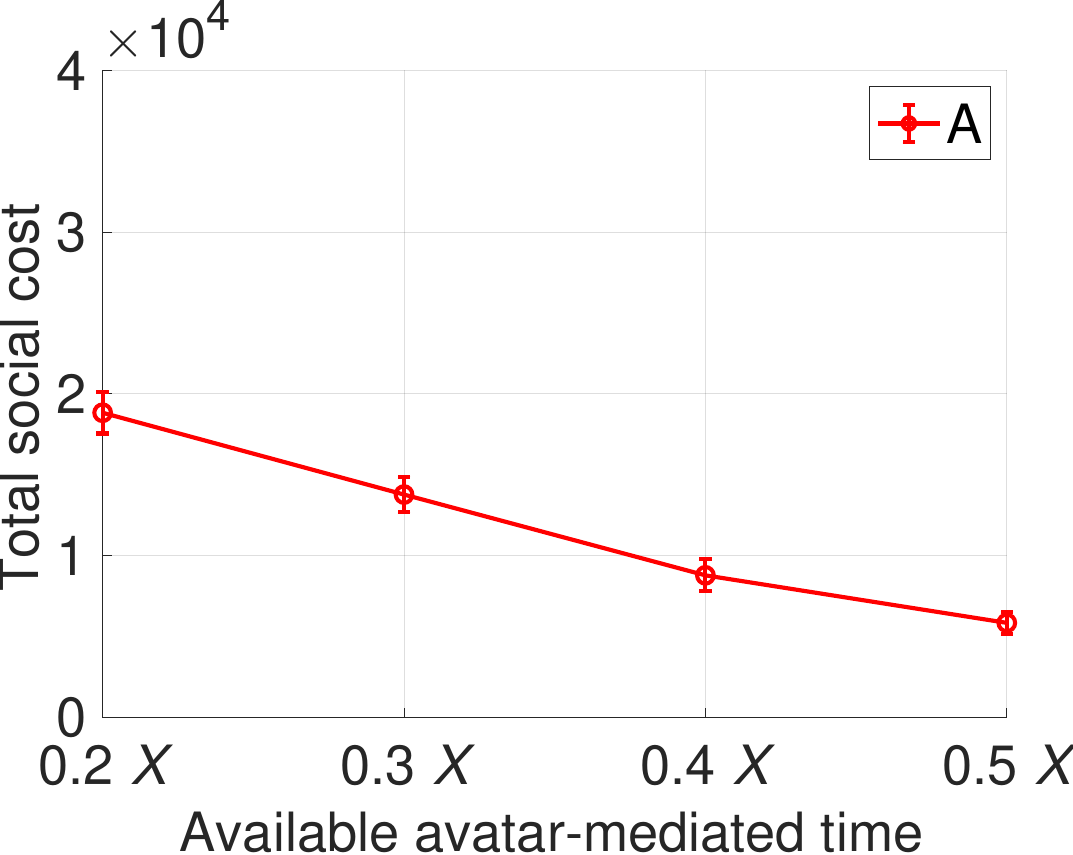}
        \label{fig:metrics-sub2}
    } \hspace{2mm}  
    \subfloat[Varying $\gamma$.]{%
        \includegraphics[width=0.3\textwidth]{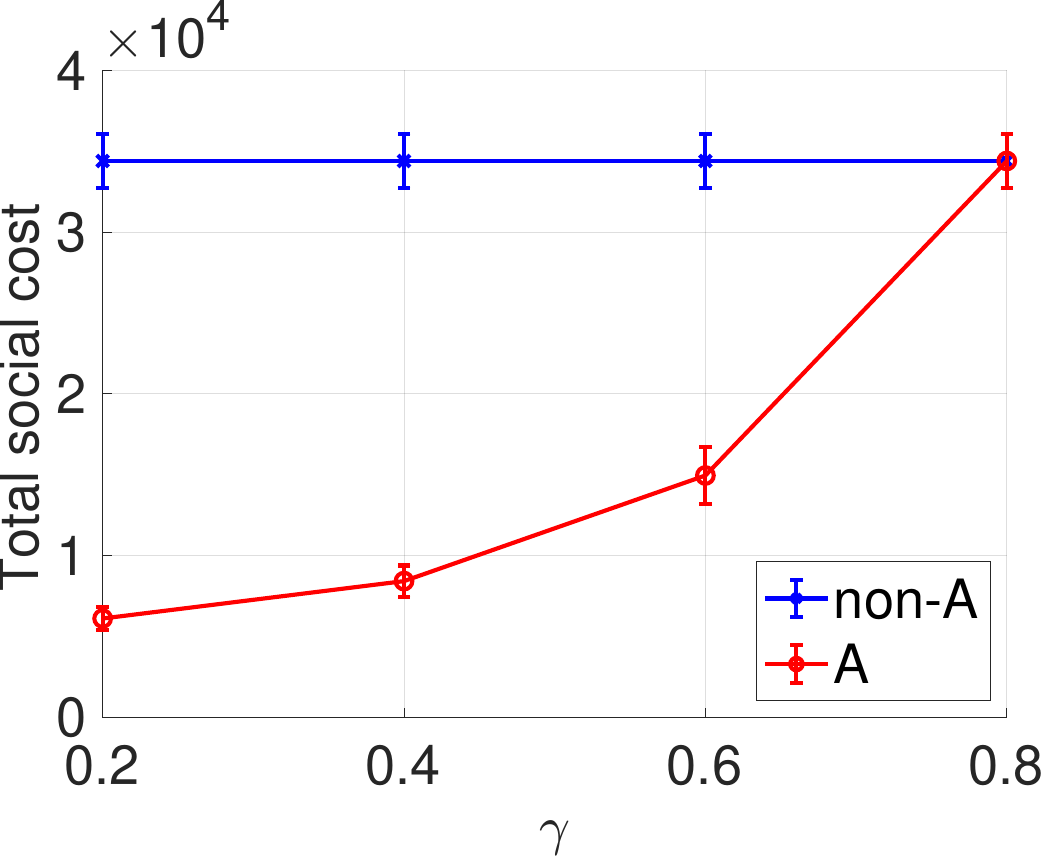}
        \label{fig:metrics-sub3}
    } 
    \caption{Varying parameters for $n = 126$.}
    \label{fig:metrics}  \vspace{-20pt}
\end{figure*}

Simulations were conducted in MATLAB (v. R2022b) to compare two distinct scenarios: the non-avatar case (non-A), where users interact without any virtual proxy, and the avatar case (A), where users can leverage the avatar by using our heuristic algorithm to assist in social interactions. All simulation parameters are fixed, in general, according to the values reported in Table~\ref{tab:exp-param}. We set $\tilde{X}' = \tilde{X}$. During the experimental campaign, we vary selected parameters to analyze their influence on system behavior. Specifically, we explore variations in the number of conflicts, the available avatar time, the maximum deadline limit, and the parameter $\gamma$. The impact of these variations is evaluated in terms of social cost, which serves as our primary performance metric and reflects the overall efficiency and coordination achieved under each scenario.

\begin{table}[b!]
    \centering
    \caption{Caption}
    \label{tab:exp-param}
    \begin{tabular}{r|l}
         Parameter & Value  \\\hline
        ego network size $|V|$ & $\{68, 126, 170\}$ alters  \\\hline
        avatar time $Y$ & $\{0.25, 0.5, 0.75, 1\} \cdot X $ hours  \\\hline
         $\beta$ & 1.29   \\\hline
         $\gamma$ & \{0.2, 0.4, 0.6, 0.8\}  \\\hline
         \# of conflicts $|E_c|$ & $\{0.2, 0.4, 0.6, 0.8\} \cdot n (n-1)/2$   \\\hline
         $Z_{max}$ bound & 304 hours \\\hline
         slot size $s(d_i)$ & 8 hours \\\hline
         $k$ & 364 days   \\\hline
         max deadline size $\max (d^{''j}_v - d^{'j}_v)$ & $\{0.1, 0.2, 0.3, 0.4 \} \cdot k$   \\\hline
    \end{tabular}
\end{table}

\subsection{Simulation results and discussion}

We now compare the non-avatar and avatar-mediated coordination cases across various experimental conditions. The analysis demonstrates the robustness and efficiency of avatar mediation in reducing social cost and improving scalability under diverse constraints.


Fig.~\ref{fig:varyConflicts} illustrates how social cost varies with the number of conflicts across different ego network sizes. In the non-avatar case, social cost increases sharply as conflicts rise. In contrast, the avatar-mediated case exhibits a more gradual increase, leading to a widening absolute gap in total cost between the two approaches. While this absolute difference grows with conflict density, the relative improvement (measured as the percentage reduction in social cost) tends to decrease. Specifically, the percentage advantage drops from approximately 95\% to 75\% as the number of conflicts increases. Nevertheless, this decreasing trend in relative gains is roughly consistent across all network sizes considered (68, 126, and 170), indicating that avatar-based mediation continues to provide significant benefits even under heavy load. Importantly, this suggests that avatars are especially valuable in avoiding cost escalation, although the marginal benefit flattens as size challenges scale up.

Fig.~\ref{fig:varyConflicts-histograms} shows the distribution of per-alter social cost. The avatar case clearly demonstrates a left-shifted distribution, indicating that a larger number of alters are served with lower individual costs. In contrast, the non-avatar distribution is broader and more symmetric, with a significant portion of alters experiencing moderate to high costs. This suggests that avatars are effective in flattening the social cost landscape, enabling a more equitable and optimized allocation of resources. In practical terms, this means more participants benefit from cost-efficient outcomes when avatars are involved in the socialization process.

In Fig.~\ref{fig:metrics}, we further examine system sensitivity under three key parameters: deadline strictness, available avatar time, and debriefing efficiency $\gamma$. As shown in Fig.~\ref{fig:metrics-sub1}, stricter deadlines substantially increase the social cost in the non-avatar case, reflecting limited room for socialization and higher missed opportunities. The avatar case, however, exhibits a less steep cost curve, showcasing its ability to absorb tighter constraints more gracefully by efficiently redistributing socialization load. Fig.~\ref{fig:metrics-sub2} depicts the relation between allocated avatar time and total social cost. Initially, increasing avatar time yields steep reductions in social cost, indicating high marginal utility. However, once avatar availability reaches approximately 50\% of the time budget used in the non-avatar scenario, the marginal benefit begins to plateau. This saturation point suggests that only a moderate investment in avatar time is needed to reap most of the benefits, offering an efficient trade-off between socialization and resource utilization. Finally, as shown in Fig.~\ref{fig:metrics-sub3}, the network performs well with avatar mediation for $\gamma$ values up to 0.6, showing a clear performance gain over the non-avatar case. However, as $\gamma$ approaches 0.8, the advantage diminishes, and performance converges. This implies that when the debriefing efficiency is low, the added value of the avatar becomes marginal, as natural convergence emerges without being able to introduce explicit mediation.


\section{Conclusions} \label{sec::conc}

This work proposed a computational framework for integrating independent avatars into Metaverse-based social networks, combining anthropological grounding with realistic constraints like social conflicts and time limits. After proving the NP-hardness of the problem, we introduced a heuristic to balance feasibility and cost. Simulations show that avatar mediation reduces social cost, improves fairness, and remains effective under strict deadlines and limited avatar time, defining a practical efficiency frontier. These findings support the role of avatars as scalable proxies for time-constrained users. Future work will explore learning-based coordination, trust dynamics, and heterogeneous avatar behaviors.

\ifCLASSOPTIONcaptionsoff
  \newpage
\fi




\balance
\bibliographystyle{IEEEtran}
\bibliography{IEEEabrv,bibtex/refs}
\end{document}